\documentclass[12pt,a4paper]{article}
\usepackage[utf8]{inputenc}
\usepackage{amsmath}
\usepackage{amsfonts}
\usepackage{amsthm}
\usepackage{amssymb}
\usepackage{authblk}
\usepackage{xcolor}
\usepackage[hidelinks]{hyperref}
\hypersetup{colorlinks,linkcolor={red!50!black},citecolor={blue!50!black},urlcolor={green!50!black}}

\theoremstyle{plain}
\newtheorem{thm}{thm}
\newtheorem{lem}{lem}
\newtheorem{pro}{pro}
\newtheorem{cor}{cor}
\theoremstyle{plain}
\newtheorem{defn}{defn}
\theoremstyle{plain}
\newtheorem{rem}{rem}

\newcommand{\citep}{\cite}
\newcommand{\F}{{\mathbb{F}}}
\newcommand{\C}{\mathcal{C}}
\newcommand{\D}{\mathcal{D}}
\newcommand{\A}{\mathbf{A}}
\newcommand{\Fq}{\mathbb{F}_q}
\renewcommand{\S}{\mathfrak{S}}
\renewcommand{\H}{\mathbf{H}}
\newcommand{\Om}{\Omega}
\newcommand{\rank}{\mathbf{rank}\;}
\newcommand{\PP}{\mathbb{P}}
\newcommand{\G}{\mathbf{G}}
\newcommand{\x}{\mathbf{x}}
\newcommand{\Fqm}{\mathbb{F}_{q^m}}
\newcommand{\y}{\mathbf{y}}
\renewcommand{\d}{\mathbf{d}}
\newcommand{\dr}{d^R}
\newcommand{\M}{\mathbf{M}}
\newcommand{\GLn}{GL_n}
\newcommand{\<}{\left<}
\renewcommand{\>}{\right>}
\renewcommand{\AA}{\mathbb{A}}
\newcommand{\qS}{\mathfrak{S}_q}
\newcommand{\0}{\mathbf{0}}
\newcommand{\N}{\mathbf{N}}
\newcommand{\s}{\mathbf{s}}
\newcommand{\w}{\mathbf{w}}
\newcommand{\B}{\mathbf{B}}
\renewcommand{\l}{\mathbf{\lambda}}
\newcommand{\m}{\mathbf{\mu}}
\newcommand{\e}{\epsilon}
\newcommand{\HC}{\mathcal{H}}
\renewcommand{\L}{\Lambda}

\author{Tovohery Randrianarisoa\thanks{Email: tovo@aims.ac.za - The author is supported by the Swiss SNF grant No. 181446.}}
\affil{Department of Mathematics, IIT Bombay}

\title{A Geometric Approach to Rank Metric Codes and a Classification of Constant Weight Codes}

\begin{document}

\maketitle

\begin{abstract}
In this work we develop a geometric approach to the study of rank metric codes. Using this method, we introduce a simpler definition for generalized rank weight of linear codes. We give a complete classification of constant rank weight code and we give their generalized rank weights. 
\end{abstract}

\section{Introduction}\label{sec:1}
In his paper \citep{Wei91}, Wei has introduced new parameters for linear block codes to characterize the performance of linear codes when such codes are used on the wiretap channel of type II. These parameters, called the Generalized Hamming Weights, were already studied earlier in \citep{HKM77} in a different context. Later, Cai and Yeung introduced in \citep{CY02} an equivalent scheme for secure network coding. 
Silva et al. considered this problem where they introduced the use of rank metric codes. 
Several works on generalized weight of rank metric codes appeared after this \citep{OS12,Duc15,KMU15,JP15,MP16,Rav16,JP17}. 
In these works, multiple notions of generalized rank weights were proposed. 
And ultimately, these definitions appeared to be equivalent. 
Continuing with these works, we consider a newer but more natural definition of generalized rank weights for rank metric codes. 
Our definitions are analogous to the definitions given by Wei in \citep{Wei91} for Hamming metric. 
Furthermore, we consider a geometric approach analogous to the work of Tsfasman and Vladut in \citep{TV95}. Such approach to study linear codes were already done in \cite{Hil78} and it was probably introduced much earlier.
Many results in this paper are translation of the definitions and results from \citep{Wei91,TV95} from the Hamming metric to the rank metric. 
This approach is done by using geometrical sets which we call $q$-system. It turns out that $q$-systems are the vectorial counterpart of linear sets \cite{Pol10,LV15} and although we did not know about them, they appear to be well studied geometric objects, there are many works about them and recently, results about the connection between linear sets and rank metric codes were presented \cite{Lun17,She16,CMPZ17,CMPZ19,SV19}.
Going back to Hamming metric codes, in their work, \citep{LC09}, Liu and Chen give some properties of constant weight linear codes. Another result of Bonisoli \citep{Bon83} also gives a characterization of constant weight linear codes.
These results give us a similar idea for the main result of this work, which is a complete classification of constant rank weight codes.

In Section \ref{sec:2}, we first recall some results in the Hamming metric setting, for us to see the analogy when we present our results for rank metric codes. 
In Section \ref{sec:3}, we introduce the notion of generalized rank weights for rank metric codes, both in analogy with the work in \citep{TV95} and \citep{Wei91}. 
In Section \ref{sec:4}, we follow the description of wiretap network codes from \citep{OS12} to show why our definition of generalized rank weights is proper for applications. 
In Section \ref{sec:5}, we will give some properties of the generalized rank weights. For instance we will see the monotonicity and the duality properties for these parameters.
In Section \ref{sec:6}, we recall the notion of linear sets and we explain why they are the projective version of the notion of $q$-system. We give a brief summary of the relations between linear rank metric codes and linear sets. 
In Section \ref{sec:7}, we will give a classification of constant rank weight codes. 
In fact if the dimension is at least $2$, then, up to equivalence, there is only one constant rank weight code. We give the construction for such constant rank weight codes.

Before we start let us define the following notations.

\begin{itemize}
\item For a field $\F$, $\F^{m\times n}$ denotes the set of all $m\times n$ matrices over $\F$.
\item For a matrix $\A$, $\A^T$ is its transpose.
\item For two matrices $\A$ and $\B$ of the appropriate size, $[\A|\B]$ is the concatenation of the two matrices columnwise.
\item For an $\F$-linear code (resp. $\F$-vector space) $\C$, $\D<\C$ means that $\D$ is an $\F$-subcode (resp. $\F$-subspace) of $\C$.
\item $\dim_\F V$ is the dimension of $V$ as an $\F$-vector space.
\item If $\A$ is a matrix over $\F$ with $n$ columns, then for $I\subset \{1,\cdots, n\}$, $\A_I$ denotes the matrix obtained from $\A$ by keeping only the $i$-th columns with $i\in I$.
\item $\{\A\}$ denotes the set vectors formed by the columns of $\A$.
\item For a field $\F$ and an $\F$-vector space $X$, the matrix $\G^X$ over $\F$ denotes a generator matrix of $X$ i.e. its rows form an $\F$-basis of $X$.
\item Conversely, for a field $\F$ and a matrix $\A$, $\<\A\>_\F$ denotes the $\F$-vector space generated by the rows of $\A$.
\item $\{\A\}_{\F} = \<\A^T\>_{\F}$ denotes the $\F$-vector space generated by the columns of $\A$.
\item For a set of vectors of the same length $X$, $[X]$ denotes a matrix with the elements of $X$ as columns (after fixing their order).
\item For an $n$ dimensional $\Fq$-subspace $X$ of $\Fqm^k$ (considered as a vector space over $\Fq$), we choose an arbitrary basis $\{P_1,\ldots,P_n\}$ of $X$ and $[X]_{\Fq}$ denotes the matrix with $P_1,\ldots,P_n$ as columns.
\end{itemize}

\section{Hamming metric codes}\label{sec:2}
Many results in this paper will use a generalization of the notion of projective system into the rank metric setting. So before we proceed to the rank metric codes, it is natural to recall the geometric approach for linear codes by Tsfasman and Vladut in \cite{TV91}. First let us recall the definitions of generalized weights as it was introduced by Wei in \citep{Wei91}.

Let $\Fq$ be a finite field with $q$ elements. Suppose $\C$ is an $[n,k]$-linear code over $\Fq$. For a non-zero $\D<\C$, we define the support $\S(D)$ of $\D$ as
\[
\S(\D) = \left\lbrace i\colon \exists (x_1,\cdots,x_n)\in \D, x_i\neq 0 \right\rbrace.
\]
As we mentioned in the introduction, the notion of generalized Hamming weights were introduced by Wei in \cite{Wei91} and they have some applications in cryptography, with the use of codes in wire-tap channels of type II \cite{OW85}.

\begin{defn}[Generalized Hamming weights]\label{defn:1}
For an $[n,k]$-linear code $\C$ over $\Fq$ and any integer $r$ with $1\leq r\leq k$, the $r$-th generalized weight of $\C$, denoted by $d_r(\C)$, is 
\[
d_r(\C) = \min \left\lbrace |\S(\D)|\colon \D<\C , \dim_{\Fq} \D = r \right\rbrace.
\]
\end{defn}

It is easily seen that the minimum distance of a linear code $\C$ is given by $d_1(\C)$. Namely, a non-zero codeword defines a subcode of dimension $1$ and the size of the support is just the number of non-zero elements in the codeword.

An alternative description of the generalized Hamming weight was given in Theorem 2 of \cite{Wei91} by the next proposition. For us to see the application of the generalized rank weights, we will need a similar definition. For now, let us see it for the Hamming metric.

\begin{pro}\label{pro:1}
Let $\H$ be a parity check matrix of an $[n,k]$-linear code $\C$ over a field $\Fq$. For a subset $I\subseteq \lbrace 1,\cdots, n\rbrace$, let $\H_I$ be the submatrix of $\H$ corresponding to $I$. Then, for $1\leq r\leq k$,
\[
d_r(\C) = \min\left\lbrace |I|\colon |I| - \rank \H_I \geq r\right\rbrace.
\]
\end{pro}

Another definition of the generalized Hamming weight was also given by Tsfasman and Vladut in \cite{TV95} using a geometric approach. Before that we need to translate the notion of linear code into some geometric terms.
Furthermore, when we talk about linear code, we will always talk about {\em non-degenerate} linear codes i.e. no columns of any fixed generator matrix of the code is the zero column.

\begin{defn}[Projective system]\label{defn:2}
A projective system over $\Fq$ with parameters $[n,k,d]$ is a set $X$ of $n$ points (not necessarily distinct) in a $(k-1)$-dimensional projective space $\PP = \PP^{k-1}(\Fq)$ such that $X$ is not contained in any hyperplane in $\PP$ and
\[
n-d = n-d(X) := \max\left\lbrace |X\cap H|\colon H \text{ a hyperplane in }\PP \right\rbrace,
\]
where the intersection is counted with multiplicity.
\end{defn}

It was shown in \cite{TV91} that, up to equivalence, $[n,k,d]$-projective systems  are in one to one correspondence with non-degenerate $[n,k,d]$-linear codes. For the definition of the equivalence used in this correspondence, one can have a look at \cite{TV91}.

The definition of the minimum distance $d$ can naturally be generalized to the generalized weights of a projective system:
\[
n-d_r(X) := \max\left\lbrace  |X\cap \Pi|\colon \Pi \text{ a projective subspace of codimension } r \text{ in } \PP \right\rbrace.
\]
Obviously, we have $d_1(X) = d(X)$. As it was shown in \cite{TV95}, we have the following proposition.
\begin{pro}
Let $\G$ be the generator matrix of a linear code $\C$ and let $X = \{\G\}$ (i.e. the set of vectors formed by the columns of $\G$) be its corresponding projective system. Then $d_r(\C) = d_r(X)$.
\end{pro}

Now the goal of the next section is to generalize these notions in the rank metric setting.

\section{Rank metric codes}\label{sec:3}
In this section, we give the analogy to the geometric approach of Tsfasman and Vladut. Rank metric codes were independently introduced by Gabidulin \cite{Gab85} and Delsarte \cite{Del78}. This class of codes are very interesting as they have found applications in cryptography and network coding. Before we define the analogue of projective system let us briefly recall some properties of rank metric codes.

For a vector $\x = (x_1,\ldots,x_n) \in \Fqm^n$ the rank weight, $\rank \x$, of $\x$ is the dimension of the $\Fq$-subspace of $\Fqm$ generated by $\{ x_1,\ldots, x_n\}$.
For two vectors $\x = (x_1,\ldots,x_n)$ and $\y=(y_1,\ldots,y_n)$ in $\Fqm^n$, the distance between $\x$ and $\y$ is
\[
\d(\x,\y) = \rank(\x-\y).
\]

An $[n,k,d]$ $\Fqm$-linear rank metric code $\C$ over the extension $\Fqm/\Fq$ is an $\Fqm$-subspace $\C<\Fqm^n$ of dimension $k$ such that the minimum of the rank distance between two distinct codewords is $d$. The minimum rank distance of a rank metric code $\C$ will be denoted by $\dr(\C)$. If $\G$ is a generator matrix of a linear rank metric code $\C$, then, as in \cite{JP17}, we say that $\C$ is {\em non-degenerate} if the columns of $\G$ are linearly independent over the field $\Fq$.

For an $[n,k,d]$ $\Fqm$-linear rank metric code $\C$ over the extension $\Fqm/\Fq$, the Singleton bound states that $d\leq n-k+1$ \cite{Gab85,Del78}. If such bound is attained, i.e. $d= n-k+1$, then we say that the code is a {\em maximum rank distance} (MRD) code. MRD codes are interesting because if we fix $\Fqm,n,d$, then they are the codes with the largest possible size. That property allows cryptographers to use smaller key sizes when they use MRD codes in public key cryptosystems based on linear codes \cite{GPT91}. Many constructions of MRD codes exists. See for example \cite{Gab85,Del78,She16,She19}.

\begin{rem}
We defined rank metric codes as subspaces of $\Fqm^n$, however they can also be defined by subspaces of linearized polynomials as in \cite{She16,She19}. Furthermore, rank metric codes can also be considered to be linear over $\Fq$ only, in this case codewords can also be represented by matrices in $\Fq^{m\times n}$. If $\C\subset \Fq^{m\times n}$ is linear over $\Fq$ only and its minimum distance is $d$, then the Singleton bound is given by $\dim_{\Fq} \C\leq \max\{m,n\}(\min\{m,n\}-d+1)$. As usual, codes satisfying the equality is called maximum rank distance codes.
\end{rem}
\begin{rem}
The reader should pay attention to the fact that the vector space $\Fqm^n$ can be equipped both with the rank metric and the Hamming metric. In fact, depending on the situation, we use both metrics on the same codeword.
\end{rem}

From now on, we will only consider \textbf{non-degenerate} linear rank metric codes. 
We have the following equivalent definition of the rank weight of a codeword.

\begin{pro}
Let $\C$ be a non-degenerate linear rank metric code over the extension $\Fqm/\Fq$.
The rank weight of a vector $\x\in \C$ is equal to the minimum of the Hamming weight of $\x\M$, where $\M$ runs through $\GLn(\Fq)$. 
\end{pro}
\begin{proof}
Suppose that $\x$ is a codeword with minimum rank weights $l$ and, up to permutation, we may assume that $\x = (x_1,\dots,x_n)$, where $\rank \x = \rank (x_1,\dots,x_l)$. Then we can find an invertible matrix $\M$ such that $\x\M = (x_1,\dots,x_l,0,\dots,0)$.
\qed
\end{proof}

It is this fact that helps us to generalize all notions from Hamming metric codes to rank metric codes. For instance we can define the minimum distance of a rank metric code as follows.

For a linear code $\C$ of length $n$ and a matrix $\M\in \Fq^{n\times n}$, $\C\M$ denotes the linear code such that all codewords are products $\x\M$ for some $\x\in \C$. 

\begin{thm}\label{thm:1}
The minimum distance $\dr(\C)$ of a rank metric code $\C$ is equal to
\[
\dr(\C) = \min_{\M\in \GLn(\Fq)} d(\C\M),
\]
where on the right hand side we have the Hamming distance.
\end{thm}

We call two $[n,k,d]$-linear rank metric codes $\C_1$ and $\C_2$ equivalent if there exists $\M\in \GLn(\Fq)$ and $a\in \Fqm^*$ such that $\C_2 = a\C_1\M$. This definition follows the notion of linear rank metric equivalence in \cite[Proposition 1]{Mor14}.

Using Definition \ref{defn:2} and Theorem \ref{thm:1}, we can define the minimum rank distance of a code $\C$ generated by a generator matrix $\G^\C$ as
\[
n - \dr(\C) = \max \left\lbrace |\{\G^\C M\}\cap H|\colon \M\in \GLn(\Fq),\;H \text{ a hyperplane in }\PP \right\rbrace.
\]
The vectors in $\{\G^\C M\}$ are now considered to be the corresponding class in $\PP = \PP^{k-1}(\Fqm)$. We can do this because, since we only consider non-degenerate rank metric code, then $\{\G^\C M\}$ does not contain the zero vector. 

Now, in the above equation, if $H$ is a hyperplane such that $n-\dr(\C) = |\{\G^\C \M\}\cap H|=l$, then $l$ columns of $\G^\C\M$ are in $H$.
Thus, working in the projective setting, the $\Fq$-projective subspace generated by these columns are in $H$ since $H$ is a hyperplane in $\PP$ and can also be considered as an $\Fq$-projective subspace of $\PP$.
Since $\M$ runs through all the possible invertible matrices over $\Fq$, then this leads us to think of the minimum number of $\Fq$-linearly independent elements of $\{\G^\C\}_{\Fq}\cap H$ instead of $|\{\G^\C \M\}\cap H|$, for all $\M\in \GLn(\Fq)$. In other words, we may think of the notion of dimension of some vector space over $\Fq$.
Recall that $\{\G^\C\}_{\Fq}$ denotes the $\Fq$-vector space generated by the columns of $\G^\C$.

Now, we are ready to formalize this generalization with the notion of projective system. However, for simplicity, there is no need for us to work in the projective space. We will work in the affine space. We will come back to the projective setting in a later section with the notion of linear sets.

\begin{defn}[$q$-Systems]\label{defn:4}
Let $\Fqm/\Fq$ be an extension of finite degree $m$. An $[n,k,d]$ $q$-system over $\Fqm$ is an $n$-dimensional $\Fq$-subspace $X$ of the $k$-dimensional affine space $\AA = \Fqm^k$ such that $X$ is not contained in any hyperplane in $\AA$ and
\[
n-d := \max\left\lbrace \dim_{\Fq} X\cap H\colon H \text{ a hyperplane in }\AA \right\rbrace.
\]
$d$ is called the minimum distance of $X$ and it is usually denoted by $\dr(X)$.
\end{defn}

\begin{rem}\label{rem:1}
Note that in the above definition, the affine space $\AA = \Fqm^k$ is defined over $\Fqm$ whereas $X$ is considered to be only an $\Fq$-subspace.
\end{rem}

We can see that an $[n,k,d]$ $q$-system over $\Fqm$ can be defined as a set $X$ of $n$ points in a $(k-1)$-dimensional projective space $\PP = \PP^{k-1}(\Fqm)$ over $\Fqm$ such that for any $\M\in \GLn(\Fq)$, $\{[X]\M\}$ is a projective system over $\Fqm$. 

Two $q$-systems $X_1$ and $X_2$ are called equivalent if there is a vector space automorphism $\phi$ of $\AA$ such that $\phi(X_1)=X_2$.

Similarly to linear Hamming metric codes, $[n,k,d]$ $q$-systems are in one to one correspondence with non-degenerate $[n,k,d]$-linear rank metric codes.

Namely we have the following proposition.
\begin{thm}\label{thm:2}
Let $\Fqm/\Fq$ be an extension of degree $m$. The equivalence classes of $[n,k,d]$ $q$-systems are in one to one correspondence with the equivalence classes of non-degenerate $[n,k,d]$-linear rank metric codes via the correspondence
\[
X  \leftrightarrow \C = \<[X]_{\Fq}\>_{\Fqm},
\]
or equivalently
\[
X = \{\G\}_{\Fq}\leftrightarrow \C = \<\G\>_{\Fqm}.
\]
\end{thm}
\begin{proof}
It is easy to see that equivalent $q$-systems give equivalent linear rank metric codes. Let us check the parameters.

Let $X = \< P_1,\ldots,P_n \>_{\Fq}$ be an $[n,k,d]$ $q$-system. Let $\C$ be the linear code such that the columns of the generator matrix $\G$ are the $P_i$'s. It is obvious that the length of the code is $n$. For any $\x \in \AA$, $x P_i^T\neq 0$ for some $i$ with $1\leq i\leq n$. Otherwise such $x$ defines a hyperplane which contains all the $P_i$'s. Therefore the rows of $\G$ are linearly independent over $\Fqm$, thus the dimension of the code is $k$. The equality of the minimum distance comes from the definition of $d$ for the $q$-system and from Theorem \ref{thm:1} for the rank metric codes. Since the $P_i$'s are linearly independent over $\Fq$, then $\C$ is non-degenerate. 

One can easily check that this map is surjective by taking $X$ as the vector space generated over $\Fq$ by the columns of the generator matrix $\G$ of a rank metric code.
\qed
\end{proof}

We are now ready to define the generalized weights of a $q$-system.

\begin{defn}[Generalized rank weight]\label{defn:5}
Let $X$ be a $q$-system over $\Fqm$. The generalized weights of a $q$-system is given by
\begin{align*}
n-\dr_r(X) := \max\left\lbrace \right.& \dim_{\Fq} X \cap \Pi: \\
& \left. \Pi \text{ an } \Fqm\text{-subspace of codimension } r \text{ in } \AA\right\rbrace.
\end{align*}
\end{defn}

We easily see that the minimum distance $\dr(X)$ of a $q$-system $X$ is given by $\dr_1(X)$.

\begin{defn}\label{defn:6}
Let $\C$ be an $[n,k,d]$-linear rank metric code with generator matrix $\G$. We define the generalized rank weights of $\C$ as the generalized rank weights of the $q$-system $X$ generated over $\Fq$ by the columns of $X$, i.e., $X = \{\G\}_{\Fq}$.
\end{defn}

In fact, this definition of generalized weights of $\C$ has an analogous version from Definition \ref{defn:1}. First let us define the notion of $q$-support for rank metric code.
\begin{defn}[$q$-Support]\label{defn:7}
Let $Y$ be a vector space with generator matrix $\G^Y$ over $\Fqm$. The $q$-support $\qS(\G^Y)$ of the matrix $\G^Y$ is the $\Fq$-vector space generated by the columns of $\G^Y$. A $q$-support $\qS(Y)$ of $Y$ with respect to $\G^Y$ is $\qS(Y):= \qS(\G^Y)$. 
\end{defn}

\begin{rem}\label{rem:2}
In Definition \ref{defn:7}, since there are multiple choices for the generator matrix $\G^Y$, then there are also multiple choices for the $q$-support of $Y$. However, it is not difficult to show that $\dim_{\Fq} \qS(\G^Y)$ does not depend on the choice of the generator matrix. Therefore $\dim_{\Fq} \qS(Y)$ is uniquely defined and this also does not affect the notion of generalized weight as we define in the following theorems.
\end{rem}

The first theorem is the analogue of Definition \ref{defn:1} whereas the second theorem is the analogue of Proposition \ref{pro:1}.

\begin{thm}\label{thm:3}
Let $\C$ be an $[n,k,d]$-linear rank metric code with generator matrix $\G$, then the generalized rank weights of $\C$ are equal to
\[
\dr_r(\C) = \min \left\lbrace \dim_{\Fq}\qS(\D):\quad \D<\C , \dim_{\Fqm} \D = r \right\rbrace.
\]
\end{thm}
\begin{proof}
Suppose that the generalized rank weight in Definition \ref{defn:5} is equal to $d$ and the generalized rank weight in Theorem \ref{thm:3} is equal to $d'$. Our goal is to show that $d=d'$. We assume that $X$ is the $q$-system generated by the columns of $\G$.

From Definition \ref{defn:5}, suppose that $X_1=X \cap \Pi$, $X=X_1\oplus_{\Fq} X_2$ with $n-d = \dim_{\Fq} X_1$ and $d=\dim_{\Fq} X_2$. Assume that $\Pi^\perp$ is the orthogonal complement of $\Pi$ in $\AA$ with generator matrix $\G^{\Pi^\perp}$, and therefore it has dimension $r$. $\D = \G^{\Pi^\perp}\C$ is a subcode of $\C$ of dimension $r$. Then we have a $q$-support of $\D$ given by $\qS(\D) = \qS\left( \G^{\Pi^\perp}X_1 \oplus_{\Fq} \G^{\Pi^\perp} X_2 \right)$. Therefore $\qS(\D) = \qS\left( \G^{\Pi^\perp}X_2 \right)$, since $\G^{\Pi^\perp} X_1= \{\0\}$. Since $X_2$ does not contain any element of $\Pi$, then we have $\dim_{\Fq} \qS(\D) = \dim_{\Fq} \qS(X_2) = d$. By definition of $d'$, we must have $d'\leq d$.

Conversely, suppose that $d' = \dim_{\Fq} \qS(\D)$ such that $\D<\C$ of dimension $r$. We can write $\D = \G^{\Pi}\C$ and define $\Pi$ to be the $\Fqm$-subspace of dimension $r$ in $\AA$ with generator matrix is $\G^{\Pi}$. By the definition of $q$-support, for the $\Fq$-subspace $\G^{\Pi} X<X$, $\dim_{\Fq} \G^{\Pi} X = d'$. Suppose that $\Pi^\perp$ is the orthogonal complement of $\Pi$ in $\AA$. Then, $\Pi^\perp$ is of codimension $r$. We claim that  $\dim_{\Fq}\Pi^\perp \cap X = n-d'$ so that $n-d'\leq n-d$  i.e. $d\leq d'$, which will conclude the proof. But by hypothesis $d' = \dim_{\Fq} \G^{\Pi} X$ and $X$ is of dimension $n$, therefore there is $X_1<_{\Fq} X$ such that $\G^{\Pi} X_1 = \<\0\>$ and $\dim_{\Fq} X_1 = n-d'$. But obviously, $X_1 = \Pi^\perp \cap X$.
\qed
\end{proof}

\begin{thm}\label{thm:4}
Let $\C$ be an $[n,k,d]$-linear rank metric code with parity check matrix $\H$, then the generalized rank weights of $\C$ are equal to
\begin{align*}
\dr_r(\C) = \min \left\lbrace \right. i:\quad  & 1\leq i\leq n,\\
& \left. \exists \M\in \Fq^{n\times i},\; \rank \M = i,\; i - r \geq \rank \H\M \right\rbrace.
\end{align*}
\end{thm}
\begin{proof}
Suppose that the generalized rank weight in Theorem \ref{thm:3} is equal to $d$ and the generalized rank weight in Theorem \ref{thm:4} is equal to $d'$. Our goal is again to show that $d=d'$. 

Suppose that $\M\in \Fq^{n\times d'}$ and $d'=\H\M$ such that $d'-r\geq \rank \H\M$.
We may assume that $\M\in \Fq^{n\times d'}$ such that $\rank \M = d'$ and $d'-r =\rank \H\M$. Indeed, if $\M\in \Fq^{n\times d'}$ such that $\rank \M = d'$ and $d'-r >\rank \H\M$, then we may remove a column of $\M$ to get a matrix $\A$ of rank $d'-1$ in $\Fq^{n\times (d'-1)}$ and $d'-1-r \geq\rank \H\M\geq \rank \H\A$. This is in contradiction with the definition of $d'$ as being the minimum.

Consider the $\Fqm$-linear map 
\begin{align*}
(\H\M)^T:\Fqm^{d'}&\rightarrow \Fqm^{n-k}\\
 (x_1,\dots,x_{d'})&\mapsto (x_1,\dots,x_{d'})(\H\M)^T.
\end{align*} 
Let $U$ be the kernel of the above map. Since, $\rank \H\M = d'-r$, by the rank nullity theorem, $\dim_{\Fqm} U = r$. Now, let $\D$ be the subspace of $\Fqm^n$ defined by 
\[
\D = \left\lbrace \x\in\Fqm^n, x_{d'+1} = \cdots = x_n = 0 \text{ and } (x_1,\cdots,x_{d'})\in U \right\rbrace.
\]
 Thus, $\dim_{\Fq}\qS(\D)\leq d'$ and $\D$ has dimension $r$. 
 Furthermore $\D[\M|\N]^T$ has dimension $r$, where $\N$ is some matrix to concatenate with $\M$ so that $[\M|\N]$ is invertible. 
 $\D[\M|\N]^T$ is also a subcode of the code $\C$ since $\D[\M|\N]^T \H^T = U\M^T \H^T = U(\H\M)^T = \{\0\}$. By the definition of $d$ in Theorem \ref{thm:3}, we have $d\leq \dim_{\Fq}\qS(\D)\leq d'$. 

Conversely, let $\D$ be a subcode of dimension $r$ of $\C$ with $\dim_{\Fq}\qS(\D)=d$. Thus there is an invertible matrix $\M$ such that $\S(\D \M) = \{1,\cdots,d\}$ (notice that here we have the classical support in the Hamming metric setting). Let $\M_s$ be the matrix consisting of all the first $d$ columns of $\M$. $\M_s$ is of rank $d$. By the definition of the support, $\D\M_s$ has dimension $r$. If $\G^\D$ is the generator matrix of the subcode $\D$, then $\G^\D \M \M^{-1} \H^T = \0$, where $\H$ is the parity check matrix of the code. Therefore, if we write $\M=[\M_s|\N]$, then 
\[
\G^\D [\M_s|\N] \M^{-1} \H^T = \0, \text{where } \M = [\M_s|\N].
\]
However, $\G^\D\N=\0$ so that
\[
[\G^\D\M_s|\0] \M^{-1} \H^T = \0.
\]
Set $\M'$ to be the matrix consisting of the first $d$ rows of $\M^{-1}$. Therefore $\G^\D\M_s \M' \H^T = \0$. Since $\G^\D\M_s$ is of rank $r$, then its kernel (as a linear map $\Fqm^d\rightarrow \Fqm^r$) has dimension $d-r$. Therefore $\rank \M'\H^T\leq d-r$, since the column space of $\M'\H^T$ is in the previous kernel. Hence $d-r \geq \rank \H(\M')^T$, and $\rank (\M')^T = d$. Thus, by the definition of $d'$, $d'\leq d$. This concludes the proof.
\qed
\end{proof}

We can also check that $\dr_1(\C)$ corresponds to the original definition of the minimum rank distance of a code.

The definition using the $q$-system notion is very helpful in computing the generalized rank weights of some linear codes. 
For instance, we will see in a later section that the generalized weights of a constant rank weight code can be easily computed. 
Theorem \ref{thm:3} is especially useful to obtain the definition of generalized rank weights with Theorem \ref{thm:4}. 
Theorem \ref{thm:4} in turn is needed to see why our notion of generalized rank weight characterizes the performance of codes when used in wiretap network codes as we will see in a later section. There are several approaches for the notion of generalized weights for rank metric codes \cite{JP17,KMU15,OS12,Duc15,Rav16}. These existing definitions were shown to be equivalent in \cite{JP17}.

For the remaining part of this paper, we will switch between these three definitions of generalized rank weights depending on the situation. We will use both the notions of $q$-systems and rank metric codes interchangeably, depending on which notion we find easy to write down a proof. 

A natural question to us is whether our definition of generalized rank weights is equivalent to the other known definitions in \cite{JP17,KMU15,OS12,Duc15,Rav16}. Indeed we show that our definition is equivalent to the definition in \cite{JP17}.

We fix a basis $\{b_1,\cdots,b_n\}$ of $\Fqm/\Fq$ and for $x=\sum_{i=1}^n l_i b_i \in \Fqm$, let $\overline{x} = (l_1,\cdots,l_n)\in \Fq^n$. For a codeword $\x = (x_1,\cdots,x_n)$ in $\Fqm^n$, Let $\M_B(\x)$ be the matrix such that the $i$-th column of $\M_B(\x)$ is the $\overline{x_i}$. We define the matrix support of $\x$ as the rowspace of $\M_B(\x)$. For a subspace $\D$ of a linear code $\C$, the matrix support $\S_M(\D)$ of $\D$ is defined to be the $\Fq$-vector space generated by the matrix support of each element of a basis of $\D$. Then, in \citep{JP17}, we have the following definition of generalized rank weight.

\begin{defn}\label{defn:8}
Let $\C$ be a linear rank metric code over $\Fqm/\Fq$. Then the generalized rank weight is defined as
\[
\dr_r(\C) = \min \left\lbrace \dim_{\Fq}\S_M(\D):\quad \D<\C , \dim_{\Fqm} \D = r \right\rbrace
\]
\end{defn}

If we fix a basis $\{\x_1,\cdots,\x_r\}$ of $\D$, then $\dim_{\Fq}\S_M(\D)$ is also equal to the dimension of the columnspace of 
\[
\begin{bmatrix}
\M_B(\x_1) \\
\vdots \\
\M_B(\x_n)
\end{bmatrix}.
\]
But this later dimension is also equal to the dimension of the vector space generated over $\Fq$ by the columns of
\[
\begin{bmatrix}
\x_1 \\
\vdots \\
\x_n
\end{bmatrix}.
\] 
And therefore $\dim_{\Fq}\S_M(\D) = \dim_{\Fq}\qS(\D)$. Thus the definition of generalized weight in Definition \ref{defn:8} is equal to the definition of generalized weight in Theorem \ref{thm:3}.

\section{Wiretap network codes}\label{sec:4}
We briefly explain the scheme as it was shown in \cite{OS12}. Let $\C$ be a non-zero $[n,k,d]$-linear code with parity check matrix $\H$. The secret message is a vector $\s\in \Fqm^k$. The message which is sent across the network is $\x = (x_1,\cdots,x_n)\in \Fqm^n$ randomly chosen in the coset with syndrome $\s=\H\x^T$. $\Fq$-linear combinations of the $x_i$'s will be spread across the network via known encoding. We assume that the eavesdropper, Eve, can observe $u$ edges. So, we can say that Eve knows $\w = \B \x^T$, with $\B\in \Fq^{u\times n}$. We assume that $\B$ is also of full rank $u$.  We want to minimize the information Eve can know about $\s$. The information Eve knows are $\B,\w,\H$. We will not go into the details of the information theoretical properties of the scheme but rather we give a simple algebraic argument. For more details one can have a look at \citep{RS07,OS12}.

Let $\<\B\>$ and $\<\H\>$ respectively be the $\Fqm$-subspaces of $\Fqm^n$ generated by the rows of $\B$ and $\H$. Suppose that $\y\in \<\B\>\cap\<\H\>$. Thus we can write $\y= \l \B=\m\H$, for some $\l\in \Fqm^u$ and $\m\in \Fqm^{n-k}$. Multiplying by $\x^T$ , we get a relation
\[
\m\s = \l\w,
\]
where $\s$ is the syndrome defined earlier and $\w$ is known by Eve.

Thus an element of the intersection $\<\B\>\cap\<\H\>$ gives a linear relation between the entries of $\s$.
The more the size of the intersection $\<\B\>\cap\<\H\>$ is, the more the linear relations about the elements of $\s$ are and therefore the more we know about $\s$. Thus to minimize the information accessed by Eve about $\s$, we want to minimize the intersection $\<\B\>\cap\<\H\>$ for any $\B\in \Fq^{u\times n}$. So, an important parameter to look at is
\[
\delta_u = \max_{\substack{\B\in \Fq^{u\times n}\\ \rank \B=u}} \dim \<\B\>\cap\<\H\>.
\]
We want to look at the largest possible $\delta_u$ for a particular $\H$ in order to decide if $\H$ defines the best code.

For $\B\in \Fq^{u\times n}$, let $\M\in \Fq^{(n-u)\times n}$ be a generator matrix of the orthogonal complement of $\<\B\>$ as a subspace of $\Fqm^n$. Thus $\y\in \<\B\>\cap\<\H\>$ is equivalent to $\y\M^T = \0$ and $\y\in \<\H\>$. So the dimension of $\<\B\>\cap\<\H\>$, is equal to the dimension of the kernel of the map $\<\H\> \rightarrow \Fqm^{n}$ where $\y\mapsto \y\M^T$. By the rank nullity theorem, the later dimension is equal to $(n-k)-\rank \H\M^T$. Therefore our task is equivalent to finding the minimum
\[
\Delta_u = \min_{\substack{\M\in \Fq^{(n-u)\times n}\\ \rank \M=n-u}} \rank \H\M^T.
\]
For such $\Delta_u$, there is $\M^T\in \Fq^{n\times (n-u)}$ of rank $n-u$ such that $\rank \H\M^T\leq \Delta_u$. Therefore, by Theorem \ref{thm:4}, we have 
\begin{equation}\label{eq:1}
\dr_{n-u-\Delta_u}(\C)\leq n-u.
\end{equation}

Furthermore, by Theorem \ref{thm:4} (and as we saw in its proof), there exists $\M_1\in \Fq^{n\times \dr_{n-u-\Delta_u+1}(\C)}$ such that $\rank \M_1 = \dr_{n-u-\Delta_u+1}(\C)$ and $\dr_{n-u-\Delta_u+1}(\C) - n+u+\Delta_u-1 = \rank \H\M_1$. If we suppose that $n-u\geq \dr_{n-u-\Delta_u+1}(\C)$, then $\dr_{n-u-\Delta_u+1}(\C) = n-u-\e$, $\e\geq 0$. Therefore, $\M_1\in 
\Fq^{n\times (n-u-e)}$ with $\rank \M_1 = n-u-e$ and $\rank \H\M_1= \Delta_u-e-1$. 

Now choose a matrix $\N$ over $\Fq$ such that $\M_2^T = [\M_1|\N] \in \F_q^{n\times (n-u)}$ and $\rank \M_2^T = n-u$. 
Since we added $e$ columns from $\M_1$ to get $\M_2^T$, then $\rank \H\M_2^T\leq \Delta_u-1$. Hence, by definition we have $\Delta_u\leq \Delta_u - 1$ which is a contradiction.
Therefore
\begin{equation}\label{eq:2}
n-u< \dr_{n-u-\Delta_u+1}(\C).
\end{equation}
Equations \eqref{eq:1} and \eqref{eq:2} give us the following theorem.
\begin{thm}\label{thm:5}
\[
\dr_{n-u-\Delta_u}(\C)\leq n-u< \dr_{n-u-\Delta_u+1}(\C).
\]
\end{thm}

The above proof is largely inspired by a proof of the same statement in the context of Hamming code in \cite{Wei91}. This theorem implies that the gain of information for the eavesdropper exactly occurs at the generalized weights. This makes them as interesting parameters for a code. The use of $\delta_u$ to describe the security parameters is suggested in \cite{OS12}. However the I have not seen the statement of Theorem \ref{thm:5} as I wrote it here. A different expression of the use of generalized weights as parameters for the security of wiretap network codes can also be found in \cite{KMU15}.

Since our scheme is the same as the scheme in \cite{OS12}, this confirms the fact that our definition of generalized rank weights is equivalent to existing definitions.

In the next sections, we will have a look at the properties of the generalized rank weights.

\section{Properties of generalized rank weights}\label{sec:5}

The first important properties of generalized rank weights is the monotonicity. 
The proof uses the geometric property in analogy with \citep{TV95}. Since our definition is equivalent to existing definitions, there is not really a need to present the proofs of the following properties. In fact, the monotonicity, duality of the generalized rank weights and the generalized Singleton bound were already proved but using different definitions \cite{KMU15,Duc15}. However, we still think that it is nice to give the proof of the monotonicity and duality using our definitions. Our proofs are different and they are largely inspired by \cite{Wei91}. We adapt the method therein context of rank metric codes.

\begin{thm}[Monotonicity]\label{thm:6}
Let $\C$ be a $[n,k]$-linear rank metric code over $\Fqm/\Fq$. Let $\dr_r(\C)$ be the generalized weight of $\C$, then
\[
0<\dr_1<\cdots<\dr_k = n.
\]
\end{thm}
\begin{proof}
First, we show that $d_r>0$ for any $r\leq k-1$, i.e.
\begin{align*}
\max\left\lbrace \right. & \dim_{\Fq} X \cap \Pi: \\
& \left. \Pi \text{ a subspace of codimension } r \text{ in } \AA \right\rbrace < n.
\end{align*}
By definition of $q$-system, $X$ is not contained in any hyperplane and thus not in any subspaces of codimension $i>0$. 
Therefore $\dim_{\Fq} X\cap \Pi<n$ for any $\Pi$ and $\M$.

Now, we want to show that for $1\leq r\leq k-1$, $\dr_r<\dr_{r+1}$. 
Suppose that $\dim_{\Fq} X\cap \Pi_{r+1}=n-\dr_{r+1}$, $\Pi_{r+1}$ of codimension $r+1\neq k$. 
By the first part of the proof, $n-\dr_{r+1}<n$. So, there is $P$ such that $X=\<P\>_{\Fq}\oplus_{\Fq} X_1$ such that $P\notin\Pi_{r+1}$. 
Now, take $\Pi_r = \<\Pi_{r+1},P\>_{\Fqm}$. Since the codimension of $\Pi_{r+1}$ is $r+1$, then the codimension of $\Pi_r$ is $r$. If 
\[
X\cap \Pi_{r+1} = \< P_1,\cdots,P_{n-\dr_{r+1}}\>_{\Fq},
\]
and $P\notin X\cap \Pi_{r+1}$, then
\[
X\cap \Pi_{r} = \< P_1,\cdots,P_{n-\dr_{r+1}},P\>_{\Fq}.
\]
Therefore $n-\dr_{r+1}<n-\dr_r$.

Finally, since having a codimension equal to $k$ means that the subspace is the zero space, then $\dr_k = n$. 
\qed
\end{proof}
As a consequence of the monotonicity and the Singleton bound, we have the following corollary.
\begin{cor}[Generalized Singleton bound]
Let $\C$ be a $[n,k]$-linear rank metric code over $\Fqm/\Fq$. Let $\dr_r(\C)$ be the generalized weight of $\C$, then
\[
\dr_r(\C)\leq n-k+r.
\]
\end{cor}

The next property is the duality theorem. The proof will follow the method in \cite{Wei91}.

\begin{thm}[Duality]\label{thm:7}
Let $\C$ be an $[n,k]$-linear rank metric code and let $\C^\perp$ be its dual code. Then 
\[
\{\dr_1(\C),\cdots,\dr_k(\C)\}\cup \{n+1-\dr_1(\C^\perp),\cdots, n+1-\dr_{n-k}(\C^\perp) \} = \{1,2,\cdots,n\}.
\]
\end{thm}
\begin{proof}
We claim that for $t = k+r-\dr_r(\C^\perp)$, $\dr_t(\C)\leq n-\dr_r(\C^\perp)$. 
To prove this claim, suppose that $\dr_r(\C^\perp) = \dim_{\Fq} \qS(\D)$ such that $\D<\C^\perp$ and $\dim \D = r$. 
If $\G^\D$ is the generator matrix of $\D$ and $\H$ is the parity check matrix of $\C$, then
\[
\H=
\left[
\begin{array}{c}
\G^\D \\
\hline
\H'
\end{array}
\right]
\]
By the definition of $\qS(\D)$, there is an invertible matrix $\M$ over $\Fq$ such that
\[
\H\M = 
\left[
\begin{array}{c|c}
\G_1 & \0 \\
\hline
\H_1 & \H_2
\end{array}
\right],
\]
where $\G_1 \in \Fqm^{r\times \dr_r(\C^\perp)}$. Indeed, we can choose $\G_1$ so that $\qS(\D) = \{\G_1\}_{\Fq}$ and hence we can find $\M$ such that $\G^\D \M = [\G_1|\0]$.

If we define $\M_s$ as the matrix obtained with the last $n-\dr_r(\C^\perp)$ columns of $\M$, 
then $n-k-r \geq \rank \H_2= \rank \H\M_s$ and $\M_s\in \Fq^{n\times (n-\dr_r(\C^\perp))}$ has rank $n-\dr_r(\C^\perp)$. 
Therefore, by Theorem \ref{thm:4}, $\dr_t(\C)\leq n-\dr_r(\C^\perp)$.

Next we prove that $n+1-\dr_r(\C^\perp)\neq \dr_i(\C)$ for all $i,r$ where these generalized rank weights are defined. Suppose the contrary. 
By the first part we have $\dr_t(\C)\leq n-\dr_r(\C^\perp)$. 
Thus $n+1-\dr_r(\C^\perp) = \dr_{t+j}(\C)$, $j>0$. 
By, Theorem \ref{thm:3}, there is a subcode $\D$ of $\C$ of dimension $t+j$ such that $\dim_{\Fq} \qS(\D) = n+1-\dr_r(\C^\perp)$. 
Thus, if $\G$ is the generator matrix of $\C$, then there is an invertible matrix $\M$ over $\Fq$ such that
\[
\G\M = 
\left[
\begin{array}{c|c}
\G_1 & \0 \\
\hline
\G_2 & \G_3
\end{array}
\right],
\]
where $\G_1 \in \Fqm^{(t+j)\times (n+1-\dr_r(\C^\perp))}$. 
Now define $\M_s$ to be from the last $\dr_r(\C^\perp)-1$ columns of $\M$. 
Thus $\rank \G\M_s=\rank \G_3\leq k-t-j$ such that $\M_s\in \Fq^{n\times (\dr_r(\C^\perp)-1)}$ and $\M_s$ has rank $\dr_r(\C^\perp)-1$. Again, by Theorem \ref{thm:4},
\[
\dr_{\dr_r(\C^\perp)-k+t+j-1}(\C^\perp)\leq \dr_r(\C^\perp)-1,
\]
i.e.
\[
\dr_{r+j-1}(\C^\perp)\leq \dr_r(\C^\perp)-1.
\]
This is in contradiction with the monotonicity in Theorem \ref{thm:6}.
\qed
\end{proof}

\section{Linear sets}\label{sec:6}
Linear sets are well known sets in the area of geometry. They generalize the concept of subgeometry of a projective space. They were used to construct blocking sets \cite{Lun99} and they were extensively studied. One can for example see \cite{Pol10,LV15} and the references therein. Recently, relations between rank metric codes and linear sets were studied, especially for MRD codes. In this section we give a summary of the notion of linear sets and we give relations between them and rank metric codes.

\begin{defn}
Let $\Om = PG(V,\Fqm) = \PP^{r-1}(\Fqm)$ be a projective space. A set $L$ of points in $\Om$ is called an $\Fq$-linear set of $\Om$ of rank $n$ if it is given by all the non-zero vectors of an $n$-dimensional $\Fq$-vector subspace $X$ of $V$ i.e.
\[
L = L_X:= \{\<\x\>_{\Fqm}\colon \x\in X\backslash\{\0\}.
\]
\end{defn}

Hence a linear set is just set of points defined by a $q$-system $X$. This already allows us to construct a linear set from a rank metric code. If $X$ has dimension $n$ over $\Fq$ then $L_X$ is said to have rank $n$. If $\L = PG(W,\Fqm)$ is a subspace of $\Om$ then $L_X\cap \L = L_{W\cap X}$ is also a linear set.

\begin{defn}\label{defn:10}
Let $L_X$ be an $\Fq$-linear set of $\Om$ of rank $n$ and define $\L = PG(W,\Fqm)$ as a subspace of $\Om$ of dimension $r$. We say that $\L$ has weight $w_{L_X}(\L)$ with respect to $L_X$ if $\dim_{\Fq}(W\cap X)=w_{L_X}(\L)$ i.e. if $L_{W\cap X}= \L\cap L_X$ has rank $w_{L_X}(\L)$.
\end{defn}

The weight of a linear set can be used to define the generalized weights of a rank metric code.

\begin{thm}
Let $\C$ be an $[n,k]$-linear code over $\Fqm/\Fq$ and suppose that $\G$ is a generator matrix of $\C$. Let $X$ be the $q$-system defined by the columns of $\G$ and define the $F_q$-linear set $L_X$ in $\Om = PG(V,\Fqm) = PG(k-1,\Fqm)$. Then the $r$-th generalized weights of $\C$ satisfy
\[
n - d_r(\C) = \max \{ w_{L_X}(\L)\colon \L \text{ is a subspace of codimension } r \text{ of } \Om \}.
\]
\end{thm}
\begin{proof}
This follows directly from the definition of generalized weights in Definition \ref{defn:5} and the notion of weight of subspaces in Definition \ref{defn:10}.
\qed
\end{proof}

\begin{rem}
In the first version of this paper, we were not aware of the notion of linear sets. Only after a reviewer told us about this notion, we believe that it is worth it to give the above relation between rank metric codes and linear sets. In fact, following the approach of Tsfasman and Vladut, linear sets are the $q$-analogue of the projective systems. In this regards, we may also call linear sets as {\em projective $q$-systems}.
\end{rem}

As we mentioned at the beginning of this section, linear sets and rank metric codes were already shown to be related \cite{Lun17,She16,CMPZ17,CMPZ19,SV19}. We give some of this correspondence.

\begin{defn}
An $\Fq$-linear set $L_X$ of $\Om$ or rank $n$ is scattered if all of its points have weight $1$. It is called a maximum scattered $\Fq$-linear set if it is of highest possible rank.
\end{defn}

In \cite{She16}, it was shown that maximum scattered $\Fq$-linear sets of $PG(1,q^n)$ correspond to $\Fq$-linear MRD code. Notice that the linearity of the rank metric code here is as an $\Fq$-vector space, whereas we only consider $\Fqm$-linear rank metric codes when we worked on the $q$-system. Furthermore, \cite{CMPZ17} shows that MRD codes can be constructed from every scattered linear set of rank $rm/2$ of $PG(r-1,q^m)$ where $rn$ is even. 

The construction from \cite{CMPZ17} is as follows. 

Let $X$ be an $(rm/2)$-dimensional $\Fq$-subspace of $V=V(r,q^m)$, $r$ even, and let $i=\max\{\dim_{\Fq}(X\cap \<v\>_{\Fqm}:v\in V\}$. Let $G:V\rightarrow W$ be an $\Fq$-linear function , with $W=V(rm/2,q)$ such that $\mathrm{Ker}\,G = U$.

Define $\C(X,G) = \{G\circ \tau_v\colon v\in V\}$, where $\tau_v\colon \lambda\in \Fqm \mapsto \lambda v\in V$.

\begin{thm}[\cite{CMPZ17}]
If $i<n$, then $\C(X,G)$ is an $\Fq$-linear rank metric code of dimension $rm$, dimension $m$ and minimum distance $m-i$. Moreover, $\C(X,G)$ is an MRD-code if and only if $L_X$ is a maximum scattered $\Fq$-linear set.
\end{thm}

A further study of this correspondence can be found in \cite{SV19}, where they give a geometric interpretation. We would like to point out that the correspondence, in Lemma 2.2 of that paper, between linear sets and rank metric codes is similar to the relation between linear sets from a $q$-system and the corresponding rank metric code at the beginning of this section. The equivalence classes of rank metric codes and linear sets were studied in \cite{SV19}. The connection between the rank weight of a linear code and weight of hyperplanes with respect to a linear set were also given. With the $q$-system approach, we describe the higher rank weights, i.e. the generalized rank weights of a linear codes. And as we will see in the next section, the $q$-system approach allows us classify constant weight rank metric codes.

\section{Constant rank weight codes}\label{sec:7}
In this section, we show that the geometric approach helps studying rank metric codes. 
In particular we can easily classify constant rank weight codes.
First we want to show the following lemma which is useful to characterize constant rank weight codes.
\begin{lem}\label{lem:1}
Let $X\subset \Fqm^k$ be a $q$-system of parameters $[n,k,d]$. 
Suppose that there is an integer $l$ such that for any $\Fqm$-subspace $S$ of $\Fqm^k$ of dimension $r$, $\dim_{\Fq} S\cap X = l$. Then
\[
q^n = \left|\Fqm^k\cap X\right| = (q^l-1)\frac{q^{mk}-1}{q^{mr}-1} + 1.
\]
\end{lem}
\begin{proof}
We follow a method in \cite{LC09}. Define a value function on $\Fqm^k$ by
\[
v(\x) = 
\begin{cases}
1& \text{if } \x\in X,\\
0& \text{else}.
\end{cases}
\]
and extend it to any subset $S\subset\Fqm^k$ by $v(S) = \sum_{x\in S}v(S)$. 
Notice that $v(\0)=1$.  Let $L_r$ be the number of $r$-dimensional $\Fqm$-subspaces of $\Fqm^k$. 
Finally, for any fixed point $p\in \Fqm^k\backslash\{\0\}$, let $L_{r,1}$ be the number of $r$-dimensional $\Fqm$-subspaces of $\Fqm^k$ containing $p$. 
Then, it is easy to show that
\[
L_r = \frac{(q^{mk}-1)(q^{mk}-q^m)\cdots(q^{mk}-q^{m(r-1)})}{(q^{mr}-1)(q^{mr}-q^m)\cdots(q^{mr}-q^{m(r-1)})},
\]
and
\[
L_{r,1} = \frac{(q^{mk}-q^m)(q^{mk}-q^{2m})\cdots(q^{mk}-q^{m(r-1)})}{q^{m(r-1)}(q^{m(r-1)}-1)(q^{m(r-1)}-q^m)\cdots(q^{m(r-1)}-q^{m(r-2)})}.
\]
Let $S_1,\cdots,S_{L_r}$ be all the $r$-dimensional $\Fqm$-subspaces of $\Fqm^k$. Then
\begin{equation}\label{eq:3}
\sum_{i=1}^{L_r} v(S_i) = q^l L_r.
\end{equation}
Since any non-zero elements of $\Fqm^k$ appears exactly in $L_{r,1}$ $\Fqm$-subspaces of dimension $r$ and $\0$ appears in each subspaces, then
\[
\sum_{i=1}^{L_r} v(S_i) =L_{r,1} v\left(\Fqm^k\backslash\{\0\}\right) + L_r.
\]
Therefore,
\begin{equation}\label{eq:4}
\sum_{i=1}^{L_r} v(S_i) =L_{r,1} v\left(\Fqm^k\right) + L_r - L_{r,1}.
\end{equation}
Combining Equations \eqref{eq:3} and \eqref{eq:4}, we get our result.
\qed
\end{proof}

Let $\C$ be an $[n,k,d]$-linear rank metric code over $\Fqm/\Fq$. 
Recall that a constant rank weight code is a linear code such that all non-zero codewords have the same rank weight. 
If $k=1$, then it is obvious that $\C$ is a constant rank weight code. 
Thus for the remaining part of this section, we assume that $k>1$. 
Suppose that the generator matrix of $\C$ is $\G$. Let $X$ be the $q$-system corresponding to $\C$, i.e. $X$ is an $\Fq$-subspace of $\AA=\Fqm^k$. 
Suppose that $\G^\D$ is a generator matrix of an $r$-dimensional subcode $\D<\C$. Then a generator matrix for $\D$ is $\G^\D = \M_D \G$, with $\M_\D\in \Fqm^{r\times k}$.
Define 
\[
S_\D = \{ \x\in \AA:\quad \M_\D\x = \0\}.
\]
Then $\dim_{\Fqm} S_\D = k-r$. In fact this relation gives a one-to-one correspondence between subspaces of $\AA$ of dimension $k-r$ and subcodes of $\C$ of dimension $r$. 
Moreover,
\[
n-\dim_{\Fq} S_\D \cap X= \dim_{\Fq} \qS(\D). 
\]

Since we have a constant rank weight code, then $\dim_{\Fq} \qS(\D)=d$, for any subcode of dimension $1$ of $\C$. Therefore, by the above correspondence,
$\dim_{\Fq} S \cap X = n-d$ for any hyperplane $S$ of $\Fqm^k$. Now, we choose $l=n-d$ and $r=k-1$.
By Lemma \ref{lem:1},
\begin{equation}\label{eq:5}
q^n\left(q^{m(k-1)}-1\right) = q^{mk+l}-q^l-q^{mk}+q^{m(k-1)}.
\end{equation}
We have the following properties.
\begin{itemize}
\item $1<k\leq n\leq mk$,
\item $0<l = n-d< n$.
\end{itemize}
If $l<m(k-1)$, then $l< mk$ and Equation \eqref{eq:5} gives
\[
q^{n-l}\left(q^{m(k-1)}-1\right) = q^{mk}-1-q^{mk-l}+q^{m(k-1)-l}.
\]
But then $q$ divides the LHS but not the RHS of the equation. Thus by contradiction, $l\geq m(k-1)$. However, if $l> m(k-1)$, then
\[
q^{n-m(k-1)}\left(q^{m(k-1)}-1\right) = q^{m+l}-q^{l-m(k-1)}-q^{m}+1.
\]
Since $q$ does not divide the RHS, then $n=m(k-1)$. But then $l>n$ which is contrary to $l< n$. So at the end
\[
l=m(k-1).
\]
But with Equation \eqref{eq:5}, this implies that
\[
q^n(q^l-1) = q^{mk}(q^l - 1).
\]
Since $l>0$, then $n=mk$. So, in fact $X = \AA=\Fqm^k$. In the following, we show that we indeed have a constant rank weight code for which some parameters are studied.

A particular class of linear codes in the Hamming metric are the class of Hadamard codes. 
These codes, for a particular dimension $k$ over $\Fq$, are constructed in such a way that all elements of $\Fq^k$ make the columns of the generator matrix. 
Taking $X=\AA=\Fqm^k$ generalize this construction in the rank metric setting and using the geometric approach we can easily compute the generalized weight of such code.

Let $\Fqm/\Fq$ be field extension of degree $m$. Let $k$ be a positive integer and Let $X=\Fqm^k$. 
Since $X$ is a vector space of dimension $k$ over $\Fqm$, then it is also a vector space of dimension $mk$ over $\Fq$. 
Let $n=mk$, then $X$ defines an $[n,k,d]$ $q$-system, which is given in the next theorem.
\begin{thm}
Let $X=\Fqm^k$ be an $[n,k,d]$ $q$-system defined as above. The generalized rank weights of $X$ are given by
\[
\dr_r(X) = mr.
\]
In other words, $d = m$.
\end{thm}
\begin{proof}
By definition 
\begin{align*}
n-\dr_r(X) := \max\left\lbrace \right.& \dim_{\Fq} X \cap \Pi: \\
& \left. \Pi \text{ an } \Fqm\text{-subspace of codimension } r \text{ in } \AA\right\rbrace.
\end{align*}
Notice that $\AA = X$ and therefore, 
\[
n-\dr_r(X) = (k-r)m.
\]
Therefore $\dr_r(X) = mr$.
\qed
\end{proof}

\begin{defn}
The linear code corresponding to the projective system $X=\Fqm^k$ is called the Hadamard rank metric code which we denote by $\HC_1(q,m,k)$. 
It has parameters $[mk,k,m]$ and  it has generalized weights $\dr_r(X) = mr$.
\end{defn}

\begin{cor}
The Hadamard rank metric code $\HC_1(q,m,k)$ is a constant rank weight code i.e. all the codewords have rank weight $m$.
\end{cor}
\begin{proof}
We have seen that $\dr(\HC_1(q,m,k))=m$. So, $\forall \x\in \HC_1,\; \rank \x\geq m$. 
But since the alphabet is over $\Fqm$, then $\rank \x\leq m$. Thus $\forall \x\in \HC_1,\; \rank \x= m$.
\qed
\end{proof}

It is interesting that this code is optimal in the sense that it reaches the bound for rank metric codes with such parameters. 
Namely for an $[n,k,d]$-linear code over $\Fqm$ with have $k\leq (n/m)(m-d+1)$ and here we have an equality. 
For a proof of such bound, one can view the code as $\Fq$-linear code where the codewords are matrices (see \citep{Del78} for example). 
Notice also that this code is linear over $\Fqm$ but not only over $\Fq$.

Taking the dual, we have the following.
\begin{defn}
The Hamming rank metric code $\HC_2(q,m,k)$ is the dual of $\HC_1(q,m,k)$.
\end{defn}
Using the duality from Theorem \ref{thm:7}, we get the following property of Hamming rank metric codes.
\begin{thm}
The Hamming rank metric code $\HC_2(q,m,k)$ has parameters $[mk,(m-1)k,2]$. Moreover the generalized weight hierarchy is given by
\[
\left\lbrace n+1-i: \quad 1\leq i<km,\; m \nmid i \right\rbrace.
\]
\end{thm}
Having a minimum distance $2$, the code $\HC_2(q,m,k)$ is not really of a particular interest for error correcting as it can only detect error of rank $1$. However, the generalized weights can be useful.

To conclude this section, we present the following classification of non-degenerate constant weight linear rank metric codes.
\begin{thm}\label{thm:10}
Let $\C$ be an $[n,k,d]$-non degenerate linear code over $\Fqm/\Fq$.
\begin{enumerate}
\item If $k=1$, then $\C=\<(a_1,\cdots,a_n)\>_{\Fqm}$ such that $\rank (a_1,\cdots,a_n) = d$.
\item If $k>1$, then $n=mk$, $\dr_r(\C) = mr$ and the columns of the generator matrix $\G$ of $\C$ is made of a basis of $\Fqm^k$ as a vector space over $\Fq$.
\end{enumerate}
\end{thm}

\begin{rem}\label{rem:3}
If the linear code $\C$ is degenerate i.e. the columns of its generator matrix $\F$ are linearly independent, 
then the code is equivalent to a linear code with generator matrix of the form $[\G'|\0]$ where, $\G'$ defines a non-degenerate rank metric code $\C'$. 
Thus we can also use Theorem \ref{thm:10} on $\C'$ in order to classify degenerate constant weight linear rank metric code $\C$.
\end{rem}

\section{Conclusion}
In this work, we considered a geometric approach of linear rank metric codes via the notion of $q$-systems which are similar to linear sets. 
We have redefined the notion of generalized rank weight and we gave new proofs of some of their properties. 
The method also helps us to completely classify constant rank weight codes. We give a construction of such codes. 
These codes are analogous to the Hadamard codes in the Hamming metric setting. 
As a future work, we want to explore the properties of rank metric codes using this geometric approach. For instance we want to study the generalized weight of $q$-cyclic rank metric codes as it was similarly studied for cyclic Hamming metric codes. We want to use the projective setting with linear sets to find linear codes whose generalized weights can be easily computed using the geometric approach.
We also want to generalize this geometric approach into rank metric codes in the Delsarte setting \cite{Del78} i.e. we want to consider rank metric codes as subspaces of matrices.

\subsection*{Acknowledgments}
I would like to thank Rakhi Pratihar and Prof. Sudhir Ghorparde for their valuable comments and suggestions on this work. I also thank the anonymous reviewer who introduced me to the notion of linear sets.

\bibliography{references}   % name your BibTeX data base

\end{document}